\documentclass[a4paper, twocolumn, accepted=2022-10-09]{quantumarticle}

\pdfoutput=1
\usepackage{amsmath}
\usepackage{amssymb}
\usepackage{epsfig}
\usepackage{comment}
\usepackage{amsthm}
\newtheorem{theorem}{Theorem}

\newcommand{\bra}[1]{\langle #1|}
\newcommand{\proj}[1]{\ket{#1}\bra{#1}}
\newcommand{\op}[2]{\ket{#1}\bra{#2}}
\newcommand{\ket}[1]{|#1\rangle}
\newcommand{\braket}[2]{ \langle #1 | #2 \rangle}
\newtheorem{lemma}{Lemma}
\newcommand{\floor}[1]{\left\lfloor #1 \right\rfloor} 
\newcommand{\ceil}[1]{\left\lceil#1 \right\rceil} 

\begin{document}

\title{Probability in many-worlds theories}

\author{Anthony J. Short}\email{tony.short@bristol.ac.uk}
\affiliation{H.H. Wills Physics Laboratory, University of Bristol, Tyndall Avenue, Bristol, BS8 1TL, U.K.}

\begin{abstract}
We consider how to define a natural probability distribution over worlds within a simple class of deterministic many-worlds theories. This can help us understand the typical properties of worlds within such states, and hence explain the empirical success of quantum theory within a many-worlds framework. We give three reasonable axioms which lead to the Born rule in the case of quantum theory, and also yield natural results in other cases, including a many-worlds variant of classical stochastic dynamics. 
\end{abstract}

\maketitle

\section{Introduction}

Despite the amazing empirical success of quantum theory, its implications for the nature of reality remain controversial. From a realist perspective, the key issue is to replace the measurement postulates of textbook quantum theory with a more objective and well-defined structure. Supplementing the theory with hidden variables \cite{deBroglie, Bohm52}, or including  spontaneous collapse  laws \cite{GRW, Pearle, CSL} are two possible approaches. However, arguably the simplest approach, initially proposed by Everett \cite{Everett}, is to drop the measurement postulates altogether. The theory then simply describes a vector in Hilbert space undergoing unitary evolution. Considering the kinds of unitary interaction that constitute measurements, we find that they can be described by a branching of the initial state  into a superposition of  `worlds' each containing a different measurement result. This is the many-worlds interpretation of quantum theory \cite{Everett, DeWitt}. Observers in each world will see a distinct result, hence on a qualitative level this  is consistent with our experiences. However, a  key challenge is  to recover the probabilistic predictions of quantum theory within this deterministic many-worlds setting.  In particular,  we want to explain why the observed relative frequencies of  quantum results in our world, for a huge variety of different scenarios, are very close to those predicted by the Born rule \cite{Born} (i.e. $\textrm{relative frequency} \approx |\textrm{quantum amplitude}|^2 $).

The approach we pursue here is to argue for a `natural'   way to pick one world at random from  a many-worlds state. i.e. a natural probability distribution over worlds. 
We can then consider which worlds are typical with respect to this probability distribution. If worlds  containing relative frequencies in quantum experiments consistent with the Born rule are typical, this would offer an explanation for our observations - that we are living in a typical world, or are a typical instance of ourself. A probability distribution over worlds can also be interpreted as a measure, allowing one to discuss properties of `most' or `almost all' worlds. The theory predicts that  atypical worlds also exist, in which other instances of ourself have seen strange results, and perhaps do not believe quantum theory. However, this seems comparable to the situation in standard quantum theory, where we explain the observed measurement results by noting that they are typical with respect to the Born rule, and concede that it is possible we could have obtained very strange results instead, and perhaps arrived at a different theory. 

 We identify three reasonable axioms, not specific to quantum theory, which we would expect  a natural probability distribution over worlds to obey. In the case of quantum theory, we find that the unique probability distribution obeying those axioms is  that given by the Born rule.  If the axioms we require for naturalness are convincing, this could therefore explain the empirical success of the Born rule. 

As the probabilistic axioms are intended to be general in nature, we  consider a simple class of many-worlds theories which includes other possibilities in addition to quantum theory.  We consider four specific  models within this framework, and highlight  any differences in the resulting probability distributions, as well as one case in which a natural probability distribution obeying our axioms does not exist. 

In order to focus solely on the issue of probability, rather than on precisely how to decompose the state into different worlds (often known as the preferred basis problem), we consider a toy-theory  in which the set of worlds are given. Intuitively, we can think of these worlds as states which could be understood classically on a macroscopic scale (e.g. tables and pointers have well-defined locations, and  observers have well-defined experiences). These states are important in relating the content of the theory to our experiences, and are related to the locality of physical interactions and decoherence \cite{Wallace10a}. 

Previous approaches to understanding probability in  many-worlds include Everett's original work \cite{Everett}, which also sought to define a natural measure over worlds in order to consider a typical observer, although Everett's assumptions (that the measure depends only on the corresponding amplitude, and is additive when several  orthogonal states are relabelled as a single state) are not completely intuitive. Typicality is also considered in a very different approach based on matter density \cite{Allori}. Vaidman discusses many different approaches to deriving the Born rule in \cite{vaidman20}. He has proposed the `measure of existence of a world' \cite{vaidman1} and notes that deriving the Born rule requires additional assumptions, that may be based on symmetries \cite{Vaidman12} or other natural properties \cite{vaidman20}. Zurek \cite{Zurek1, Zurek2} uses envariance to derive the Born rule (which involves symmetries when the state is decomposed into a system and environment). An alternative approach is based on relative frequencies for infinitely repeated experiments \cite{Everett, Hartle}, but this is challenging to relate to finite experiments. More recently, an approach based on decision theory was proposed by Deutsch and developed by Wallace and others  \cite{deutsch99, Wallace03, Greaves07, Wallace10b, GreavesMyrvold}.  Several of the technical steps in our proofs are similar to those in this decision-theoretic approach. However this approach has led to criticisms \cite{Kent10, Albert}, and because of its focus on the future rather than the past, it seems that a decision-theoretic picture alone cannot  explain why we find ourselves in a world in which the historical record of quantum experiments agrees with the Born rule. Very recently, Saunders has proposed an approach \cite{Saunders21} based on dividing the state into equal-amplitude branches based on the decoherence structure and following similar counting arguments to those in statistical mechanics.

\section{A simple class of many-worlds theories} \label{sec:theories} 

Consider a Hilbert space with a countably infinite orthonormal basis of possible `worlds' $\ket{n}$ labelled by non-negative integers, $n \in \{0,1,2,\ldots\}$. The  space may be  over the real or complex numbers depending on the particular theory under consideration.  A many-worlds state is a vector $\ket{v}$ in this space. We will refer to $v_n = \braket{n}{v}$   as the \emph{amplitude} of world $n$. For simplicity, we restrict ourselves to \emph{linear many-worlds theories} for which the evolution over a finite interval is given by a linear operator $T$, such that  
\begin{equation} 
\ket{v'} = T \ket{v},
\end{equation} 
 We will denote the matrix elements of $T$ in the basis of worlds by $T_{ij} = \bra{i} T \ket{j}$. Our aim will be to derive  a natural probability distribution $p_n$ over worlds in the state $\ket{v}$. Similarly, we will denote the probability distribution over worlds in $\ket{v'}$ by $p_n'$. Different theories will be characterised by giving the set of allowed state vectors and transformations in the theory (with the requirement that the allowed transformations must map allowed states into allowed states). The cases we consider here are:

\begin{enumerate}

\item \textbf{Quantum many-worlds theory}. The allowed states are all complex vectors satisfying $\braket{v}{v}=1$, and the allowed transformations are all unitary operators (for which $T T^{\dagger} =  T^{\dagger} T = I$, where I is the identity transformation). We will show that in this case $p_n = |v_n|^2$. 

\item \textbf{Unnormalised quantum many-worlds theory}. An alternative version of the quantum many-worlds theory with unnormalised states. The allowed states are all complex vectors satisfying $0<\braket{v}{v}< \infty$, and the allowed transformations are all unitary operators. We will show that  in this case $p_n = \frac{ |v_n|^2}{\sum_k  |v_k|^2}$. Note that in this case the probability $p_n$ corresponding to a given world depends on all amplitudes, and not only on $v_n$.\\

\item \textbf{Stochastic many-worlds theory} This represents a many-worlds version of a classical probabilistic world. The allowed states are those with real amplitudes satisfying $v_n \geq 0 \; \forall n$ and $\sum_n v_n =1$, and the allowed transformations are those satisfying $T_{ij} \geq 0\; \forall i,j$ and $\sum_i T_{ij} =1 \; \forall j$. Note that although $v_n$ and $T_{ij}$ obey the same mathematical properties as probability distributions and stochastic maps respectively, we are not assuming that they represent probabilities. However, we will show that in this case $p_n = v_n$. 

\item \textbf{Discrete many-worlds theory} A many-worlds theory in which there are an integer number of copies of each world, with the dynamics transforming each world in the original state into a finite number of new worlds. The allowed states are  all real vectors for which $v_n$ is a non-negative integer for each $n$ and $ \sum_n v_n < \infty$, and the allowed dynamics are those for which $ T_{ij} $ is a non-negative integer for all $i,j$ and $\sum_i T_{ij} < \infty \; \forall j$. We will show that in this theory there is no natural probability distribution $p_n$ obeying our axioms.
\end{enumerate}

\section{Natural axioms for a probability measure}

For each of the theories considered above, we will investigate whether one can define a probability distribution over worlds $p_n$ obeying the `reasonable' axioms given below. 

\begin{enumerate} 
\item \emph{Present state dependence -} $p_n$ depends only on the present  state $\ket{v}$, and not on how that state was generated.  \label{axiom:state} 

\item  \emph{Weak connection with amplitudes -} $v_n=0$ implies that $p_n =0$. Hence only worlds with non-zero amplitudes are considered `real' components of the many-worlds state. \label{axiom:amplitude} 

\item \emph{Weak connection with transformations -} If the set of worlds can be partitioned in such a way that a transformation $T$ acts separately on each part, then $T$ will preserve the total probability of each part. This captures the intuition that probability cannot `flow' between worlds that are  uncoupled by the dynamics. Expressing this more formally,  if there exists a partition of the non-negative integers  into subsets  $\mathcal{S}_k$ such that $T_{ij} = 0$ whenever $i$ and $j$ are in different subsets, then $\sum_{n \in \mathcal{S}_k} p_n = \sum _{n \in \mathcal{S}_k} p'_n$ for all  $k$. \label{axiom:transformation} 

\end{enumerate} 
Of these axioms, \ref{axiom:state} and \ref{axiom:amplitude} both seem relatively uncontroversial, and are usually assumed without comment in  approaches to probability in many-worlds. Axiom \ref{axiom:transformation} is more subtle and powerful, but also seems a natural property in the context of transformations. We will discuss these axioms further at the end of the paper. However, we will first  show how they can be used to derive probability measures for the theories given in section \ref{sec:theories}.

\section{Deriving probability measures} 

For ease of understanding, in this section we describe informally  how axioms \ref{axiom:state} - \ref{axiom:transformation}  allow us to derive probability rules for the many-worlds theories we consider, highlighting the key steps in the proof. Detailed formal proofs of all of the results can be found in appendix \ref{app:derivations}. 

We first prove two helpful lemmas which we use in deriving the probability rules, that apply to \emph{bounded} states (those which contain only a finite number of worlds with non-zero amplitude) 

\begin{enumerate} 

\item[L1] \emph{Equal amplitudes give equal probabilities}. For bounded states in all of the theories we consider, if two worlds  have equal amplitude then they have the same probability. 

To prove this, we first consider a transformation which swaps two worlds, one of which has zero amplitude and the other has non-zero amplitude. Due to axiom \ref{axiom:transformation}, the total probability of the two  swapped worlds must stay the same, and due to axiom \ref{axiom:amplitude} the worlds with zero amplitude must have zero probability. Hence the probabilities are also swapped. 

The same argument doesn't apply  if we swap two worlds with non-zero amplitudes directly, as probabilities could move between the worlds. However, by performing a sequence of three  swaps, each of which involves one world with zero amplitude, we can swap any two worlds and their corresponding amplitudes and probabilities whilst leaving the remaining worlds  unchanged. If the amplitudes of the two swapped worlds were initially the same, then the final state will be the same as the initial one, but with the probabilities swapped. From axiom \ref{axiom:state} it then follows that the probabilities for these two worlds must be the same. 

\item[L2] \emph{Larger amplitudes cannot lead to smaller probabilities}. For bounded states in all the theories we consider except discrete many-worlds theory, if one world has a larger amplitude than another,  its probability must be at least as large.

To prove this, we perform a transformation which branches the higher amplitude world into two final worlds, one of which has the same amplitude as the  smaller amplitude world. From axiom \ref{axiom:transformation} the sum of the probabilities in the two branches must be equal to the probability of the initial higher amplitude world. However, one of the branches has the same  probability as the smaller amplitude world, hence the probability of the higher amplitude world must be at least as large.  \vspace{0.5cm} 

\end{enumerate} 

\vspace{1cm} 

Next, we use these results to derive probability rules for our many-worlds theories. We begin by considering a particularly simple class of quantum many-worlds states in which the amplitudes are square roots of rational numbers. 
\begin{equation} 
 \ket{v} = \sum_{n=0}^{N-1} \sqrt{\frac{m_n}{M}} \ket{n}
\end{equation} 
Following a  branching strategy similar to Deutsch and Wallace \cite{deutsch99, Wallace10b}, we perform a unitary transformation in which each world $\ket{n}$ evolves independently of the others into $m_n$ new worlds  with equal amplitude. 

Each world in the final state  has the same amplitude $\frac{1}{\sqrt{M}}$ and thus equal probability $\frac{1}{M}$ from L1 above. Applying axiom \ref{axiom:transformation} to the transformation tells us that the probability in each set of branches is conserved, hence we  find that the probability of world $\ket{n}$ in the initial state must be $\frac{m_n}{M}$, which is the standard quantum result.  

To extend this result to arbitrary initial states in quantum many-worlds theory, we must perform three additional steps. Firstly, if the initial state contains an infinite number of worlds with non-zero amplitudes, we begin by transforming it into a superposition of finitely many worlds in order to create some `working space' in which to apply the above results. For example, when considering the probability $p_k$ we can   perform a unitary transformation that merges all of the worlds with labels greater than $k$ into a single world, without affecting the other worlds (and hence not changing $p_k$ due to axiom \ref{axiom:transformation}).  For an arbitrary positive integer $M$, we can then write the state as 
\begin{equation} 
\ket{v'} = \sum_{n=0}^{N-1} \sqrt{\frac{ m_n + \epsilon_n}{M}} e^{i \phi_n} \ket{n},
\end{equation} 
where  $m_n$ is the largest integer less than or equal to $M |v_n|^2$, and $0 \leq \epsilon_n <1$. Next, we eliminate the phase factors $e^{i \phi_n}$ by applying a unitary which acts separately on each world, and therefore does not affect the probabilities. Finally, we perform a branching unitary, in which each initial world $\ket{n}$ is transformed into a superposition of $m_n$ new worlds with equal amplitude,  and one new world with smaller amplitude. The final state contains approximately $M$ worlds with equal amplitude $\frac{1}{\sqrt{M}}$ and at most $N$ worlds with smaller amplitude. When  $M \gg N$, the smaller amplitude worlds are  almost irrelevant (due to L2 above), and the probability associated with world $n$ in the initial state is  
\begin{equation} 
p_n \approx \frac{m_n}{M} \approx \frac{M|v_n|^2}{M} = |v_n|^2.
\end{equation} 
By considering arbitrarily large $M$, this argument can be made exact, giving the standard quantum probability rule $p_k = |v_k|^2$.

The derivation of the probability rules in unnormalised quantum theory and stochastic many-worlds theory are very similar. In the former case, the main difference is that the final state contains approximately $MX$ worlds with equal amplitude, where $X= \sum_m |v_m|^2$. This gives $p_n \approx  \frac{m_n}{MX} \approx \frac{ |v_n|^2}{X}$, and yields the expected quantum probability rule for this case, $p_n =  \frac{ |v_n|^2 }{ \sum_m |v_m|^2}$.  In the case of stochastic many-worlds theory, the steps are identical to those for quantum many-worlds theory, except without  phase factors, square roots, and absolute-values-squared. 

It is also important to note that the probability rules derived so far satisfy the axioms for all allowed states and transformations, and not only the ones considered when constructing the proof. For example, if a unitary transformation is block-diagonal in quantum theory, it commutes with the projectors onto each block, and hence conserves the total probability of each block.  

Finally, to see that no probability rule obeying our axioms is possible for Discrete many-worlds theory, consider a transformation on the state $\ket{0} + \ket{1}$ which takes $\ket{1} \rightarrow \ket{1} + \ket{2}$ whilst leaving all other worlds unchanged. Applying L1  to the initial state and using axiom \ref{axiom:transformation} gives $p_0'=\frac{1}{2}$, while applying L1  to the final state $\ket{0} + \ket{1} + \ket{2}$gives $p_0'=\frac{1}{3}$. As this leads to a contradiction, no probability rule obeying the axioms exists\footnote{A similar argument can be used to rule out a `naive branch counting' strategy in quantum many worlds theory, in which each world with non-zero amplitude is assigned equal probability \cite{wallacebook}}.

\section{Discussion}

In this section we present some additional discussion about the axioms, theories beyond quantum theory, and decoherence in our approach. 

\subsection{Axioms} 

\begin{enumerate} 
\item \emph{Present state dependence -} This is a simplifying axiom, and incorporates the fact that the state at a given time should be sufficient to make any substantive claims about it, including the typical properties of worlds within it. Arguably if historical information is important, it should form part of the state, and our framework should be extended. Furthermore, this axiom allows probabilities to be assigned to worlds in an arbitrary initial state, without needing to know how that state was generated. 

\item  \emph{Weak connection with amplitudes -} This  seems the most compelling requirement. Without this, one could simply assert that the state is irrelevant, and $p_0 = 1$ in all cases. 

\item \emph{Weak connection with transformations -} This is the most complicated  of the three assumptions, but it is hard to see a weaker way of incorporating a dependence on the dynamics of the theory. Without this, we could assign an arbitrary probability distribution over the worlds appearing in every state (for example we could always assign probability 1 to the world with non-zero amplitude having the lowest numerical label). Within quantum theory, this also fits nicely with the continuous time picture in which $T=e^{-i H t}$ for some Hamiltonian $H$, as $T$ will act separately on each partition if $H$ does. Note that in terms of the proofs, we only need this axiom to apply to a specific set of unitaries involving branching, swapping, or merging of worlds.
\end{enumerate}

An alternative to axiom 3, which is a stronger assumption but offers a nice conceptual picture is:

\begin{itemize} 
\item[3'.] \emph{Weak connection with transformations-} For every state $\ket{v}$ and transformation $T$ in the theory, there exists a conditional probability distribution $P_{i|j}$ such that 
\begin{equation} 
p_n' = \sum_m P_{n|m} \, p_m
\end{equation}   
satisfying $P_{i|j}=0$ whenever $T_{ij}=0$. This ensures that probability can only `flow' between states which are linked by the dynamics. 
\end{itemize}
The existence of a conditional probability distribution $P_{i|j}$ for each transformation of the state  supports the idea that living within an evolving many-worlds state could feel like undergoing a stochastic evolution. There is also a nice symmetry between the claim that $T_{ij}=0 \implies P_{i|j}=0$ and $v_n =0 \implies p_n=0$.

Axiom 3' is strictly stronger than the original axiom 3 as it implies it but is not implied by it. In particular when the worlds can be partitioned into subsets on which $T$ acts separately, then the conditional probability distribution $P_{i|j}$ can only redistribute probability within these subsets, and hence the total probability of each subset is preserved. However, one can imagine a trivial theory with only  one allowed state $\ket{v} = \ket{1} + \ket{2}$, and one allowed transformation $T= \op{1}{2}+\op{2}{1}$. Under  axiom 3 we could assign  an arbitrary probability $p_1$ for this state, but with axiom 3' we must have $p_1=p_2=\frac{1}{2}$. 

Because axiom 3' is stronger than axiom 3, we can derive all of the same results as before (and Lemma 1 now follows directly by permuting the two worlds in question). However, a more subtle point is whether the probability rules derived previously actually satisfy axiom 3' in all cases. For example, for all states $\ket{v}$ and unitary transformations $T$ in quantum many-worlds theory with $p_n = |v_n|^2$, can we always find a conditional probability distribution $P_{i|j}$ satisfying the requirements of axiom 3'? This is by no means trivial, but it is a consequence of results in \cite{Dieks, Aaronson, Mister} (in particular the Flow  or Schr\"odinger Theory presented in \cite{Aaronson}) that this is indeed the case. In general, for a given $\ket{v}$ and $T$ within quantum theory, there may be many different conditional probability distributions  $P_{i|j}$ satisfying axiom 3'. If conceptual importance is to be given to this quantity, it may therefore be useful to consider adding additional axioms which specify it uniquely, or properties which are independent of this choice. Similar results will apply for unnormalised quantum theory, and in the case of  stochastic many-worlds theory we can simply take $P_{i|j} = T_{ij}$. 

Overall, it is interesting that one can overlay a stochastic evolution onto a many-worlds state with natural properties, and that doing so is helpful in deriving the Born rule. 

\subsection{Non-quantum many-worlds theories}

Recovering a natural probability rule in the case of stochastic many-worlds theory shows that this approach is not specific to quantum theory. Furthermore, this theory may also be interesting in its own right. In particular, it is difficult to make sense of theories involving objective probabilities at a fundamental level \footnote{Frequentism only gives a definite prediction for infinitely many trials, which is irrelevant for practical situations,  and the principal principle \cite{Lewis} seems somewhat unsatisfying. Kent's incompressible bit string provides another alternative approach \cite{Kent10}}. This result suggests an interesting alternative possibility of treating objective probabilities as amplitudes in a many-worlds state. The fact that we find ourselves in a world which is typical with respect to the objective probabilities, and that these are helpful in subjective decision making, would then be a consequence of the natural probability distribution over worlds matching the amplitudes. 

For discrete many world theory there is no natural probability distribution over worlds obeying our axioms. Two alternatives for understanding probabilities within this theory are to violate axiom 1 or axiom 3. In the former case, we could assume some initial probability distribution, and then update this probability distribution during transformations according to 
\begin{equation} \label{eq:discrete_prob}
p_n' = \sum_m \left(\frac{T_{nm}}{\sum_{n'} T_{n'm}}\right) p_m.
\end{equation} 
This would lead to the same natural probability distribution as a Stochastic many-worlds theory for which the transition matrix elements $\tilde{T}_{nm}$ are given by the bracketed expression in \eqref{eq:discrete_prob}. This would lead to a natural flow of probability between worlds (obeying axiom 3), but the same state could yield many different probability distributions according to how it was generated.  In the latter case, we could instead drop axiom 3, and take 
\begin{equation} 
p_n = \frac{v_n}{\sum_n' v_n} 
\end{equation} 
This is a function of only the current state (obeying axiom 1), but the probability of a world can change even when the transformation acts on it like the identity. 

\subsection{Decoherence} 

Note that this derivation of the Born rule does not rely on decoherence, and indeed employs transformations  such as  permutations of worlds which would be essentially impossible to achieve in practice. However, such transformations are possible in principle, and it is helpful to  use the full strength of the theory to generate constraints upon possible probability rules. The two particular instances where this is helpful are in proving Lemma 1, and when  compressing states with support over all worlds in order to generate `working space'. An alternative for the former is to take Lemma 1 directly as an additional axiom,  perhaps motivated by symmetry, but the internal structure of the two worlds may look very different, and it seems quite strong to assume this is irrelevant in determining $p_n$. The latter could possibly be eliminated by adding a continuity axiom, but then  one would have to choose a particular distance measure on states. Decoherence also plays a key role in explaining `collapse' in realistic situations, as it becomes practically impossible to re-interfere macroscopically distinct states.

\section{Conclusions} 

If reality has a deterministic many-worlds structure, as in the Everett interpretation, then there are no objective probabilities in the theory. In this case, how should we understand the fact that  we are living in a world in which the relative frequencies of outcomes in past quantum experiments are very close to the probabilities predicted by the Born rule? Although it is consistent to say that this is a mere coincidence  (as a world with such results must exist), it would be good to give a deeper explanation of this fact. 

The approach pursued here is to establish a natural way of picking a world at random from the many-worlds state, and then observing that  with very high probability such randomly chosen worlds have the property that they agree with the Born rule. i.e. that this is a typical property of the worlds. The key is to motivate such a natural way of picking worlds at random. One apparently natural way (at least for states containing a finite number of worlds) is to simply assign equal weight to each world with non-zero amplitude, but this ignores some of the mathematical structure in the state, and it violates one of our reasonable axioms   (axiom 3). Other problems with this strategy have been discussed in \cite{Wallace10b}. 

It seems impossible to give  a completely compelling set of requirements which a  probability distribution over worlds must obey, but  we have defined three natural axioms which are sufficient to recover the Born rule in the context of quantum theory, and also give an appealing result for classical stochastic theory. These axioms are that  probabilities: only depend on the current state, are  zero for worlds which appear with zero amplitude, and cannot flow between sets of worlds which are uncoupled by the dynamics.  It would be interesting to explore  alternative possible requirements, and would also be good to reconsider issues relating to the choice of `world' basis in this context.  

Are results such as these sufficient to explain the empirical success of the Born rule? An interesting perspective is to consider a universe in which the Everett interpretation is true, described by a single unitarily evolving quantum state. Under reasonable conditions,  such a state could be described by a superposition of branching worlds, many of which contain structures which look like people. What would it be like to live as one of those people in such a universe? If it is like our own experiences then this supports the many-worlds interpretation. If it is very different then this would rule it out. The strangest situation would be if we cannot in principle say what it is like, given that we know the correct mathematical theory describing the universe.

\vspace{2cm}
\noindent
\textit{Acknowledgments.}
This work was supported by an FQXi ``physics of what happens" grant. The author is grateful to  A.Kent, J.Barrett, and Ll. Massanes for helpful discussions.

\bibliography{manyworldsbib}

\appendix

\section{Detailed proof of probabilities} \label{app:derivations} 

We begin by proving two useful  lemmas which apply to probability distributions obeying our axioms for the many-worlds theories we consider. These lemmas  focus  on \emph{bounded} states,  which have a finite number of non-zero amplitudes (such that $v_n =0$ for $n \geq N$), but the final theorems will apply to any state.   

We first show that for any theory in which all permutations are allowed transformations (which includes all of the theories in this paper), any two worlds with the same amplitude have the same probability. 

\begin{lemma}[equal amplitudes] \label{lem:equal} Consider a theory in which all permutations of worlds  are allowed transformations ($T$ is a permutation of worlds if $T_{ij} =\delta_{i, \pi[j]}$ where $\pi$ is a bijection of the non-negative integers), and a  state $\ket{v} = \sum_{n=0}^{N-1} v_n  \ket{n}$ for which $v_n = v_m$. Then  $p_n= p_m$.
 \end{lemma} 

\begin{proof} Denote by $T^{n \leftrightarrow m}$ the transformation which swaps worlds $n$ and $m$ and leaves all other worlds unchanged ($T^{n \leftrightarrow m} = \ket{n}\bra{m} + \ket{m}\bra{n} + \sum_{k\notin  \{n,m\}} \ket{k}\bra{k}$). Hence in this case $T^{n \leftrightarrow m} \ket{v} = \ket{v}$. Directly applying axiom \ref{axiom:transformation} to this transformation is not sufficient, as this doesn't tell us how probabilities change inside the $n,m$ subspace. Instead, we note that  $v_N =0$, and that $T^{n \leftrightarrow m} = T^{n \leftrightarrow N}T^{n \leftrightarrow m} T^{N \leftrightarrow m}$. For the first transformation $T^{n \leftrightarrow N}$, consider a partition of worlds in which $n$ and $N$ are in one part, and each other world is in its own part. Due to axiom  \ref{axiom:transformation}, we know that $p_n + p_N = p'_n+p'_N$, and that $p_k= p'_k$ for all $k\notin  \{n,m\}$. Furthermore, from axiom \ref{axiom:amplitude} and the fact that $v_N=v'_n = 0$, we know that $p_N=p'_n=0$. Hence $p'_N = p_n$ and $T^{n \leftrightarrow N}$  permutes the corresponding probabilities. Following a similar logic for the two subsequent transformations $T^{n \leftrightarrow m}$ and $T^{N \leftrightarrow m}$ (each of which permutes two amplitudes, one of which is zero), we find that the probabilities $p'_n$ after the sequence $T^{n \leftrightarrow N}T^{n \leftrightarrow m} T^{N \leftrightarrow m}$are  the permutation of the original probabilities by  $T^{n \leftrightarrow m}$, hence $p_n' = p_m$. However, due to axiom \ref{axiom:state} and the fact that the initial state is the same as the final state in this case, $p_n' = p_n$. Hence $p_n = p_m$.
\end{proof}

Our second lemma shows that, in all except the discrete many-worlds theory, if one world has larger amplitudes than another it cannot have a smaller probability. This is useful in deriving the probability rule for cases where the amplitudes are not related to rational numbers in a convenient way. 

 \begin{lemma}[larger amplitudes cannot lead to smaller probabilities] Consider a state  $\ket{v} = \sum_{n=0}^{N-1}  v_n  \ket{n}$ in quantum many-worlds theory, unnormalised quantum many-worlds theory, or stochastic many-worlds theory for which $|v_l| > |v_k|$.  Then $p_l \geq p_k$. \label{lem:smaller} \end{lemma}

\begin{proof}
For quantum many-worlds theory, or unnormalised quantum many-worlds theory, consider the unitary evolution which acts in the $\{\ket{l},\ket{N}\}$ subspace as  
\begin{eqnarray} 
T\ket{l } &=& \left( \frac{v_k}{v_l}\right) \ket{l} + \sqrt{1 - \left| \frac{v_k}{v_l} \right|^2 } \ket{N} \nonumber \\
 T\ket{N} &=& \sqrt{1 - \left| \frac{v_k}{v_l} \right|^2}  \ket{l} - \left( \frac{v_k}{v_l}\right)  \ket{N}, 
\end{eqnarray} 
and satisfies $T\ket{m} = \ket{m}$ if $m \notin \{l,N\}$. For stochastic many-worlds theory consider a stochastic evolution which acts on $\ket{l}$ as  
\begin{equation} 
T\ket{l } = \left( \frac{v_k}{v_l}\right) \ket{l} + \left( 1 - \frac{v_k}{v_l} \right) \ket{N} 
\end{equation} 
and satisfies $T\ket{m} = \ket{m}$ if $m \neq l$. In both cases this means that for  $\ket{v'} = T\ket{v}$ , $v_l'=v_k'$ and hence $p_l'=p_k'$ from Lemma  \ref{lem:equal}. From axiom \ref{axiom:transformation} we have that $p_N'+p_l' = p_N + p_l $ and  $p_k'=p_k$, and from axiom \ref{axiom:amplitude} we have $p_N=0$. Hence 
\begin{equation} 
 p_l = p_l' + p_N' = p_k' + p_N' = p_k + p_N'  \geq p_k.
\end{equation}  \end{proof} 

Using the above lemmas and the probability rule axioms, we now derive the appropriate probability rules for our many-worlds theories (or in the case of discrete many-worlds theory show that this is impossible). To illustrate the key idea, we first derive the probability rule for quantum many-worlds theory in the case in which all amplitudes are square roots of rational numbers,  using a similar branching strategy to Deutsch and Wallace \cite{deutsch99, Wallace10b} (although without any decision theoretic component).

\begin{theorem}[probabilities when $v_n$ are square-roots of rational numbers] Consider a state of the form $\ket{v} = \sum_{n=0}^{N-1} \sqrt{\frac{m_n}{M}} \ket{n}$ in quantum many-worlds theory, where $m_n$ and $M$ are positive integers. Then $p_n = \frac{m_n}{M} $. \label{thm:rational} \end{theorem}

\begin{proof} 
To prove this result, we branch each world in the original state into $m_n$ new worlds in the final state.  Consider a partition of the non-negative integers into the sets 
\begin{equation} 
\mathcal{S}_k = \left\{ \begin{array}{cl} \{k \} \cup \mathcal{L}_k  & \textrm{for $0 \leq k < N$} \\ \{ k \}  & \textrm{for $k \geq N+MN$} \end{array} \right. .
\end{equation}  where
\begin{equation} 
\mathcal{L}_k = \{ N + Mk, N + Mk+1, \ldots ,N + Mk+(M-1) \} 
\end{equation} 
Now consider a unitary $T$ whose matrix elements $T_{ij}$  are nonzero only when $i$ and $j$ are in the same subset $\mathcal{S}_k$, and which has the property that  
\begin{equation} \label{eq:branching} 
T\ket{n} = \frac{1}{\sqrt{m_n}} \sum_{l=0}^{m_n-1}  \ket{N+ nM+ l}  \qquad \textrm{for $0\leq n <N$}.
\end{equation} 
Then 
\begin{equation} 
T\ket{v} = \frac{1}{\sqrt{M}} \sum_{n=0}^{N-1} \sum_{l=0}^{m_n-1} \ket{N+ nM+l} . 
\end{equation}
As this is a superposition of worlds with equal amplitude it follows from Lemma \ref{lem:equal} that the probability assigned to each world in this state must be identical. Hence $p_m'=\frac{1}{M}$ for each world in this state. From axiom \ref{axiom:transformation} and equation \eqref{eq:branching} it then follows that $p_n = \frac{m_n}{N}$.  
\end{proof} 

Next, we use Lemma \ref{lem:smaller} to extend this result to the general case of arbitrary states. 

\begin{theorem}[general quantum probability rule] For an arbitrary state $\ket{v} = \sum_{n} v_n  \ket{n}$ in quantum many-worlds theory, $p_n = |v_n|^2 $. \label{thm:quantum} \end{theorem}

\begin{proof} 
We first transform the initial state (which may have infinitely many non-zero components) into a bounded state to create some working space. Consider a particular probability $p_k$. In order to determine this, we first  perform a unitary $T$ with the property that 
\begin{align} 
T \ket{n} &= \ket{ n}  & \textrm{if $n  \leq k$}  \nonumber \\
T\left( \frac{\sum_{n=k+1}^{\infty} v_n \ket{n}}{\sqrt{\sum_{m=k+1}^{\infty}  |v_m|^2}}\right) &= \ket{k+1}&, 
\end{align}   
which acts separately on each of the worlds 0 to $k$, and collectively on the rest (merging them into a single world). This gives a bounded state containing $N=k+2$ worlds with the property that $p_k'=p_k$ from axiom \ref{axiom:transformation}. Next, we pick a large non-negative integer $M$, and write the transformed state $\ket{v'}$ as 
\begin{equation} 
\ket{v'} = \sum_{n=0}^{N-1} \sqrt{\frac{ m_n + \epsilon_n}{M}} e^{i \phi_n} \ket{n} 
\end{equation} 
in which $m_n= \floor{M |v_n|^2 }$ and  $\epsilon_n = M |v_n|^2 - \floor{M |v_n|^2 }$,  where $\floor{x}$ is the floor function (the greatest integer less than or equal to $x$).  Hence each $m_n$ is an integer and $0\leq   \epsilon_n  <1$. 

We then apply the unitary transformation $T= \sum_n e^{- i \phi_n} \proj{n}$ to remove the phase factors from the state, yielding 
\begin{equation} \label{eq:unnormalisedstep} 
\ket{v''} = \sum_{n=0}^{N-1} \sqrt{\frac{ m_n + \epsilon_n}{M}} \ket{n}. 
\end{equation} 
 As this unitary acts separately on each world, it follows from axiom \ref{axiom:transformation} that $p_k''=p_k$.  Similarly to the previous case, we then perform a transformation which branches each world in the original state into $m_n$ or $m_n +1$ worlds in the final state (depending on whether $\epsilon_n = 0$) via a unitary acting on the same partition of worlds $\mathcal{S}_k$ as in  theorem \ref{thm:rational} with the property that 
\begin{equation} \label{eq:branching2} 
T\ket{n} = \frac{1}{\sqrt{m_n + \epsilon_n}} \left( \begin{array}{c} \sum_{l=0}^{m_n-1}  \ket{N+ nM+ l} \\ \quad+ \sqrt{\epsilon_n} \ket{N+nM+m_n} \end{array}\right)
\end{equation} 
for $0\leq n <N$. The final state is then 
\begin{equation} \label{eq:final_state} 
T\ket{v'''} = \frac{1}{\sqrt{M}} \sum_{n=0}^{N-1} \left( \begin{array}{c}  \sum_{l=0}^{m_n-1} \ket{N+ nM+l}   \\ \quad+ \sqrt{\epsilon_n} \ket{N+nM+m_n} \end{array} \right).
\end{equation}
This is a superposition of  at most $M$ worlds of equal amplitude $\frac{1}{\sqrt{M}}$ and at most $N$ worlds with smaller amplitude. Given lemma \ref{lem:equal} and  lemma \ref{lem:smaller}, if we denote the probability of one of the  amplitude $\frac{1}{\sqrt{M}}$ worlds by $\delta$ and consider the total probability it follows that $\delta \geq \frac{1}{M+N}$. From axiom \ref{axiom:transformation} it therefore follows that  
 \begin{align} 
 p_k &\geq m_k \delta \nonumber \\
 & \geq \frac{m_k}{M+N} \nonumber \\
 & \geq \frac{M|v_k|^2 - 1}{M+N} \nonumber \\
 & \geq  |v_k|^2 - \frac{N|v_k|^2+1}{M+N}. 
\end{align} 
As this holds for arbitrarily large $M$ (for fixed $N$ and $v_k$) it must be the case that $p_k \geq |v_k|^2 $. A similar approach could be employed to upper bound $ p_k$, but it is simpler to consider the  total normalisation of the state  $\sum_n |v_n|^2 = 1 = \sum_n p_n$. As $p_k \geq |v_k|^2 $ for each $k$, the only  possible solution is $p_k = |v_k|^2$ for all $k$. \end{proof}

This proof can be extended very straightforwardly to the case of unnormalised quantum many-worlds theory, and stochastic many-worlds theory. 

\begin{theorem}[unnormalised quantum  probability rule] For an arbitrary state $\ket{v} = \sum_{n} v_n  \ket{n}$ in unnormalised quantum many-worlds theory,  $p_n = \frac{ |v_n|^2 }{ \sum_m |v_m|^2} $.\label{thm:unnormalised} \end{theorem}

\begin{proof} Note that Lemmas \ref{lem:equal} - \ref{lem:smaller} apply identically in unnormalised quantum many-worlds theory. The proof of Theorem \ref{thm:unnormalised} is the same as that for Theorem \ref{thm:quantum} up until \eqref{eq:unnormalisedstep}. Defining $X= \sum_m |v_m|^2$ and $M'=\ceil{MX}$, where $\ceil{x}$ is the smallest integer greater than or equal to $x$,  we then define  sets $\mathcal{S}'_k$ via 
\begin{equation} 
\mathcal{S}'_k = \left\{ \begin{array}{cl} \{k \} \cup \mathcal{L}_k  & \textrm{for $0 \leq k < N$} \\ \{ k \}  & \textrm{for $k \geq N+M'N$} \end{array} \right. .
\end{equation}  where
\begin{equation} 
\mathcal{L}'_k = \{ N + M'k, N + M'k+1, \ldots ,N + M'k+(M'-1) \}.
\end{equation}
Applying a unitary which acts on the partition $\mathcal{S}'_k$ with the property that 
\begin{equation} 
T\ket{n} = \frac{1}{\sqrt{m_n + \epsilon_n}} \left( \begin{array}{c} \sum_{l=0}^{m_n-1}  \ket{N+ nM'+ l} \\ \quad+ \sqrt{\epsilon_n} \ket{N+nM'+m_n} \end{array}\right)
\end{equation} 
for $0\leq n <N$. The final state is then 
\begin{equation} 
T\ket{v'''} = \frac{1}{\sqrt{M}} \sum_{n=0}^{N-1} \left( \begin{array}{c}  \sum_{l=0}^{m_n-1} \ket{N+ nM'+l}   \\ \quad+ \sqrt{\epsilon_n} \ket{N+nM'+m_n} \end{array} \right).
\end{equation}
this state is a sum of at most $MX$ worlds with equal amplitude $\frac{1}{\sqrt{M}}$ and at most $N$ worlds with smaller amplitude. If we denote the  probability of one of the  amplitude $\frac{1}{\sqrt{M}}$ worlds by $\delta$ as before, then it follows from  lemma \ref{lem:equal} and  lemma \ref{lem:smaller} that  $\delta \geq \frac{1}{MX+N}$. Hence from axiom \ref{axiom:transformation}, 
 \begin{align} 
 p_k &\geq m_k \delta \nonumber \\
 & \geq \frac{m_k}{XM+N} \nonumber \\
 & \geq \frac{M|v_k|^2 - 1}{XM+N} \nonumber \\
 & \geq  \frac{|v_k|^2}{X} - \frac{\frac{N}{X}|v_k|^2+1}{MX+N}. 
\end{align} 
As this holds for arbitrarily large $M$ (for fixed $N$ and $\ket{v}$) it must be the case that $p_k \geq \frac{|v_k|^2}{X} $. However, as this applies to each $k$, given the total normalisation of the probabilities, the only possible solution is $p_k = \frac{|v_k|^2}{X} =  \frac{ |v_n|^2 }{ \sum_m |v_m|^2}$ for all $k$. \end{proof}

\begin{theorem}[stochastic probability rule] For an arbitrary state $\ket{v} = \sum_{n} v_n  \ket{n}$ in stochastic many-worlds theory,  $p_n =v_n $.\label{thm:stochastic} \end{theorem}

\begin{proof}We first transform the initial state into a bounded state to create some working space.  In order to determine a particular probability $p_k$, we first  perform a  transformation $T$ given by
\begin{equation}  
T \ket{n} = \left\{ \begin{array}{cl}  \ket{ n}  & \textrm{if $0 \leq n  \leq k$}  \\  \ket{ k+1}  & \textrm{if $n  > k$} \end{array} \right.
\end{equation}   
which acts separately on each of the worlds 0 to $k$, and collectively on the rest (merging them into a single world). This gives a bounded state containing $N=k+2$ worlds with the property that $p_k'=p_k$ from axiom \ref{axiom:transformation}. Next, we pick a large non-negative integer $M$, and write the transformed state $\ket{v'}$ as 
\begin{equation} 
\ket{v'} = \sum_{n=0}^{N-1} \frac{ m_n + \epsilon_n}{M}  \ket{n} 
\end{equation} 
in which $m_n= \floor{M v_n }$ and  $\epsilon_n = M v_n  - \floor{M v_n }$,  where $\floor{x}$ is the floor function.  Hence each $m_n$ is an integer and $0\leq   \epsilon_n  <1$. We then perform a transformation which branches each world in the original state into $m_n$ or $m_n +1$ worlds in the final state (depending on whether $\epsilon_n = 0$) via a stochastic transformation acting on the  partition of worlds $\mathcal{S}_k$ as in  theorem \ref{thm:rational}, 
\begin{equation} 
\mathcal{S}_k = \left\{ \begin{array}{cl} \{k \} \cup \mathcal{L}_k  & \textrm{for $0 \leq k < N$} \\ \{ k \}  & \textrm{for $k \geq N+M$} \end{array} \right. .
\end{equation}  where
\begin{equation} 
\mathcal{L}_k = \{ N + Mk, N + Mk+1, \ldots ,N + Mk+(M-1) \} 
\end{equation}
with the property that 
\begin{equation} \label{eq:branching2} 
T\ket{n} = \frac{1}{m_n + \epsilon_n} \left( \begin{array}{c} \sum_{l=0}^{m_n-1}  \ket{N+ nM+ l} \\ \quad+ \epsilon_n \ket{N+nM+m_n} \end{array}\right)
\end{equation} 
for $0\leq n <N$. The final state is then 
\begin{equation} 
T\ket{v'} = \frac{1}{M} \sum_{n=0}^{N-1} \left( \begin{array}{c}  \sum_{l=0}^{m_n-1} \ket{N+ nM+l}   \\ \quad+ \epsilon_n \ket{N+nM+m_n} \end{array} \right).
\end{equation}
This is a superposition of  at most $M$ worlds of equal amplitude $\frac{1}{M}$ and at most $N$ worlds with smaller amplitude. Given lemma \ref{lem:equal} and  lemma \ref{lem:smaller}, if we denote the probability of one of the  amplitude $\frac{1}{M}$ worlds by $\delta$ and consider the total probability it follows that $\delta \geq \frac{1}{M+N}$. From axiom \ref{axiom:transformation} it therefore follows that  
 \begin{align} 
 p_k &\geq m_k \delta \nonumber \\
 & \geq \frac{m_k}{M+N} \nonumber \\
 & \geq \frac{Mv_k - 1}{M+N} \nonumber \\
 & \geq  v_k- \frac{N v_k+1}{M+N}. 
\end{align} 
As this holds for arbitrarily large $M$ (for fixed $N$ and $v_k$) it must be the case that $p_k \geq v_k $. However, as this apples to each $k$, given the total normalisation of the state  $\sum_n v_n = 1 = \sum_n p_n$, the only possible solution is $p_k = v_k$ for all $k$.
 \end{proof}

Finally, we show that discrete many-worlds theory admits no probability measure consistent with the axioms. 

\begin{theorem}[discrete theory admits no quantum  probability] There is no probability rule consistent with the axioms for discrete probability theory \end{theorem}

\begin{proof} Consider the initial state $\ket{v} = \ket{0} + \ket{1}$ and the transformation 
\begin{align} 
T \ket{1} &= \ket{1} + \ket{2} \nonumber \\ 
T \ket{n} &= \ket{n}   \qquad \textrm{if $n \neq 1$},   
\end{align} 
giving 
\begin{equation} 
T \ket{v} = \ket{0} + \ket{1}+ \ket{2}
\end{equation} 
This transformation is consistent with a partitioning of the worlds such that worlds $1$ and $2$ are in one part, and each other world is in its own part. Hence from axiom \ref{axiom:transformation}, we have $p_0 = p_0'$, and $p_1' + p_2' = p_1$. However, discrete many-worlds theory admits all permutations  as allowed transformations, so Lemma \ref{lem:equal} also applies and worlds with equal amplitude must have equal probability. Hence $p_0 = p_1 = \frac{1}{2} $ and $p_0' = p_1' = p_2' = \frac{1}{3}$. This creates a contradiction with $p_0 = p_0'$, hence no probability rule satisfying the axioms exists. 
\end{proof}

\end{document}